\documentclass[11pt, a4paper]{article}
\usepackage{amsmath}
\usepackage{amsthm}
\usepackage{amsfonts}
\usepackage{amssymb}

\usepackage{tikz}
\usetikzlibrary{patterns}

\usepackage{authblk}
\usepackage{fullpage}
\usepackage{todonotes}
\usepackage{hyperref}
\usepackage{algorithm}
\usepackage[noend]{algpseudocode}
\usepackage{nameref}
\usepackage{cleveref}

\DeclareMathOperator*{\E}{E}
\DeclareMathOperator*{\Var}{Var}

\newcommand{\floor}[1]{\lfloor {#1} \rfloor}
\newcommand{\bfloor}[1]{\big\lfloor {#1} \big\rfloor}
\newcommand{\bbfloor}[1]{\bigg\lfloor {#1} \bigg\rfloor}
\newcommand{\ceil}[1]{\lceil {#1}\rceil}

\newcommand\eps\varepsilon
\newcommand\Z{\mathbb Z}

\newtheorem {lemma} {Lemma}[section]

\newtheorem {property} {Property}
\newtheorem {definition} {Definition}

\newtheorem {theorem}[lemma] {Theorem}

\newcommand{\unif}{\mathcal{U}}

\newcommand{\andtt}{ \mathbin{\texttt{\&}} }

\newcommand{\ls}{ \mathbin{\texttt{<\!<}} }
\newcommand{\rs}{ \mathbin{\texttt{>\!>}} }

\title{The Power of Hashing with Mersenne Primes}

\author[1]{Thomas Dybdahl Ahle}
\author[2]{Jakob Tejs B\ae{}k Knudsen}
\author[2]{Mikkel Thorup}
\affil[1]{Facebook, BARC, \textit{thomas@ahle.dk}}
\affil[2]{University of Copenhagen, BARC, \textit {\{jakn, mthorup\}@di.ku}}

\begin{document}
\maketitle

\begin{abstract}
The classic way of computing a $k$-universal hash function is to use a random degree-$(k-1)$ polynomial over a prime field $\mathbb Z_p$.
For a fast computation of the polynomial, the prime $p$ is often chosen as a Mersenne prime $p=2^b-1$.

In this paper, we show that there are other nice advantages to using Mersenne primes.
Our view is that the hash function's output is a $b$-bit integer that is uniformly distributed in $\{0, \dots, 2^b-1\}$, except that $p$ (the all \texttt1s value in binary) is missing.
Uniform bit strings have many nice properties, such as splitting into substrings which gives us two or more hash functions for the cost of one, while preserving strong theoretical qualities.
We call this trick ``Two for one'' hashing, and we demonstrate it on 4-universal hashing in the classic Count Sketch algorithm for second-moment estimation.

We also provide a new fast branch-free code for division and modulus
with Mersenne primes. Contrasting our analytic work, this code
generalizes to any Pseudo-Mersenne primes $p=2^b-c$ for small $c$.
\end{abstract}

\setcounter{tocdepth}{2}
\renewcommand{\contentsname}{}
\vspace{-3.5em}
\tableofcontents
\newpage

\section{Introduction}

\begin{figure}
	\centering
	\begin{tikzpicture}[darkstyle/.style={circle,draw,fill=gray!40,minimum size=20}]
		\newcommand*{\figb}{8}
		\draw[pattern=north west lines, pattern color=red] (0,0) rectangle (\figb,1);
		\foreach \y in {0,...,6}
		\draw (0, \y) -- (\figb, \y);
		\foreach \x in {0,...,\figb}
		\draw (\x, 0) -- (\x, 6);
		\pgfmathsetmacro{\figbthree}{\figb - 3}
		\pgfmathsetmacro{\figbtwo}{\figb - 2}
		\pgfmathsetmacro{\figbone}{\figb - 1}
		\foreach \y in {0,...,2}
		\foreach \x in {0,...,\figbthree}
		\node [draw=none] at (.5+\x,.5+\y) {1};
		\foreach \x in {\figbthree,...,\figbone}
		\node [draw=none] at (.5+\x,.5+3) {.};
		\foreach \y in {4,...,5}
		\foreach \x in {0,...,\figbtwo}
		\node [draw=none] at (.5+\x,.5+\y) {0};
		\node [draw=none] at (.5+\figb-1,.5+0) {1};
		\node [draw=none] at (.5+\figb-2,.5+0) {1};
		\node [draw=none] at (.5+\figb-1,.5+1) {0};
		\node [draw=none] at (.5+\figb-2,.5+1) {1};
		\node [draw=none] at (.5+\figb-1,.5+2) {1};
		\node [draw=none] at (.5+\figb-2,.5+2) {0};
		\node [draw=none] at (.5+\figb-1,.5+4) {1};
		\node [draw=none] at (.5+\figb-1,.5+5) {0};
	\end{tikzpicture}
	\caption{The output of a random polynomial modulo $p=2^b-1$ is uniformly distributed in $[p]$, so each bit has the same distribution, which is only $1/p$ biased towards 0.}
	\label{fig:bits}
\end{figure}

The classic way to implement $k$-universal hashing is to use a random degree $(k-1)$-polynomial over a finite field \cite{wegman81kwise}.
Mersenne primes, which are prime numbers on the form $2^b-1$, have been used to implement finite fields efficiently for more than 40 years using standard portable code \cite{carter77universal}.

The speed of hashing is important because it is often an inner-loop
bottle-neck in data analysis. A good example is when hashing is used
in the sketching of high volume data streams, such as traffic through
an Internet router, and then this speed is critical to keep up with
the stream. A running example in this paper is the classic second
moment estimation using 4-universal hashing in count sketches
\cite{charikar04count-sketch}.
Count Sketches are linear maps that statistically preserve the Euclidean norm.
They are also popular in machine learning under the name
``Feature Hashing''~\cite{moody1989fast,weinberger2009feature}.

In this paper, we argue that uniform random values from Mersenne prime fields
are not only fast to compute but
\emph{have special advantages different from any other field.}
While it is natural to consider values mod $p=2^b-1$ as ``nearly'' uniform $b$-bit strings
(see \Cref{fig:bits}),
we show that the small bias in our hash values can usually be turned into an advantage.
In particular our analysis justify splitting single hash values into two or more for a significant computational speed-up, what we call the ``Two for one'' trick.

We also show that while the $1/p$ bias of such strings would usually result in relative errors of order $n/p$ for Count Sketch, a specialized analysis yields relative errors of just $n/p^2$.
The analysis is based on simple moments, and give similar improvements for any algorithm analyzed this way.
Loosely speaking, this means that we for a desired small error can reduce
the bit-length of the primes to less than half. This saves not only
space, it means that we can speed up the multiplications
with a factor of 2.

Finally we provide a fast, simple and branch-free algorithm for division and modulus with Mersenne primes.
Contrasting our analytic work, this code generalizes to so-called Pseudo-Mersenne primes~\cite{van2014encyclopedia} of the form $p=2^b-c$ for small $c$.
Our new code is simpler and faster than the classical algorithm of Crandall~\cite{crandall1992method}.

We provide experiments of both algorithms in \Cref{sec:experiments}.
For the rest of the introduction we will give a more detailed review of the new results.

\subsection{Hashing uniformly into b bits}\label{sec:b-bit?}
A main point in this paper is that having hash values uniform in $[2^b-1]=\{0,\dots,2^b-2\}$
is almost as good as having uniform $b$-bit strings, but of course,
it would be even better if we just had uniform $b$-bit strings.

We do have the fast multiply-shift scheme of Dietzfelbinger~\cite{dietzfel96universal}, which directly gives 2-universal
hashing from $b$-bit strings to $\ell$-bit strings, but for $k>2$,
there is no such fast $k$-universal hashing scheme that
can be implemented with standard portable code.

More recently it has been suggested to use carry-less multiplication
for $k$-universal hashing into bit strings (see, e.g., Lemire
\cite{lemire2014strongly}) but contrasting the hashing with Mersenne primes,
this is less standard (takes some work to get it to run on different
computers) and slower (by about 30-50\% for larger $k$ on the computers we tested in \Cref{sec:experiments}).
Moreover, the code for different bit-lengths $b$ is quite different
because we need quite different irreducible polynomials.

Another alternative is to use tabulation based methods which are fast
but use a lot of space \cite{Siegel04,Tho13:simple-simple}, that is,
space $s=2^{\Omega(b)}$ to calculate $k$-universal hash function in
constant time from $b$-bit keys to $\ell$-bit hash values. The large
space can be problematic.

A classic example where constant space hash functions are needed is in static two-level hash functions \cite{FKS84}.
To store n keys with constant access time, you use $n$ second-level hash tables, each with its own hash function.
Another example is small sketches such as the Count Sketch \cite{charikar04count-sketch} discussed in this paper.
Here we may want to store the hash function as part of the sketch, e.g., to query the value of a given key.
Then the hash value has to be directly computable from the small representation, ruling out tabulation based methods (see further explanation at the end of \Cref{sec:count-sketch}).

It can thus be problematic to get efficient $k$-universal hashing directly into
$b$-bit strings, and this is why we in this paper analyse the
hash values from Mersenne prime fields that are much easier to generate.

\subsection{Polynomial hashing using Mersenne primes}

Before discussing the special properties of Mersenne primes in algorithm analysis, we show how they are classically used to do fast field computations, and propose a new simple algorithm for further speed-ups in the hashing case.

The definition of $k$-universal hashing
goes back to Carter and Wegman~\cite{wegman81kwise}.
\begin{definition}
	A random hash function $h:U\to R$ is $k$-universal if for $k$
	distinct keys $x_0,\ldots,x_{k-1}\in U$, the $k$-tuple
	$(h(x_0),\ldots,h(x_{k-1}))$ is uniform in $R^k$.
\end{definition}
\noindent
Note that the definition also implies the values
$h(x_0),\ldots,h(x_{k-1})$ are independent.
A very similar concept is that of $k$-independence, which has only this requirement but doesn't include that values must be uniform.

For $k>2$ the standard $k$-universal hash function is uniformly random degree-$(k-1)$ polynomial over a prime field
$\Z_p$, that is, we pick a uniformly random vector
$\vec a=(a_0,\ldots,a_{k-1})\in \Z_p^k$ of $k$ coefficients, and define
$h_{\vec a}:[p]\to[p]$,
\footnote{ We use the notation $[s]=\{0,\ldots,s-1\}$.  }
by
\[h_{\vec a}(x)=\sum_{i\in[k]}a_i x^i \mod p.\]
Given a desired key domain $[u]$ and range $[r]$ for the hash values, we pick
$p\geq \max\{u,r\}$ and define
$h^r_{\vec a}:[u]\to[r]$ by
\[h^r_{\vec a}(x)=h_{\vec a}(x)\bmod r.\]
The  hash values of $k$ distinct keys remain independent while staying as close as possible to the uniform distribution on $[r]$.
(This will turn out to be very important.)

In terms of speed, the main bottleneck in the above approach is the mod operations.
If we assume $r=2^\ell$, the mod $r$ operation above can be replaced by a binary {\sc and} (\texttt{\&}): $x \bmod r = x \andtt r-1$.
Similarly, Carter and Wegman \cite{carter77universal} used a
Mersenne prime $p=2^b-1$,\footnote{e.g., $p=2^{61}-1$ for hashing 32-bit keys or
$p=2^{89}-1$ for hashing 64-bit keys.}
to speed up the computation of the (mod $p$) operations:
\begin{equation}
	y
	\equiv y - \floor{y/2^b}(2^b-1)
	= (y\bmod 2^{b}) + \floor{y/2^b}
	\pmod {p}.
	\label{eq:Mersenne}
\end{equation}
Again allowing us to use the very fast bit-wise {\sc and} ($\andtt$) and the right-shift ($\rs$),
instead of the expensive modulo operation.

Of course, \eqref{eq:Mersenne} only reduces $y$ to an equivalent value mod $p$, not to the smallest one, which is what we usually want.
For this reason one typically adds a test ``if $y \ge p$ then $y \gets y - p$''.
We show an implementation in \Cref{alg:Mersenne} below with one further improvement:
By assuming that $p=2^b-1\geq 2u-1$
(which is automatically satisfied in the typical case where $u$ is a power
of two, e.g., $2^{32}$ or $2^{64}$)
we can get away with only doing this test once, rather than at every loop.
Note the proof by loop invariant in the comments.

\begin{algorithm}[H]
	\caption{
	For $x\in [u]$, prime $p=2^b-1\geq 2u-1$,
	and $\vec a=(a_0,\ldots,a_{k-1})\in[p]^k$,
	computes $y=h_{\vec a}(x)=\sum_{i\in[k]}a_i x^i\mod p$.
	}\label{alg:Mersenne}
	\begin{algorithmic}
		\State $y\gets a_{k-1}$
		\For{$i=q-2,\ldots,0$}
		\Comment{Invariant: $\quad y<2p$}

		\State $y\gets y*x+a_i$
		\Comment{$\quad y<2p(u-1)+(p-1)<(2u-1)p\leq p^2$}

		\State $y\gets (y\andtt p)+(y\rs b)$
		\Comment{$\quad y<p+p^2/2^b<2p$}
		\EndFor
		\If{$y\geq p$}
		\State $y\gets y-p$
		\Comment{$y<p$}
		\EndIf
	\end{algorithmic}
\end{algorithm}

In \Cref{subsec:intro-division} we will give one further improvement to \Cref{alg:Mersenne}.
In the next sections we will argue that Mersenne primes are not only fast, but have special properties not found in other finite fields.


\subsubsection{Selecting arbitrary bits}\label{sec:power-of-two}
If we had b uniform bits, we could partition them any way we’d like and get smaller independent
strings of uniform random bits. The first property of random values modulo Mersenne primes
we discuss is what happens when the same thing is done on a random value in $[2^b - 1]$ instead.

More formally, let $\mu:[2^b]\to[2^\ell]$ be any map that selects
$\ell$ distinct bits, that is, for some $0\leq j_1<\cdots<j_{\ell}<b$,
$\mu(y)=y_{j_1}\cdots y_{j_\ell}$. For example, if $j_i=i-1$, then we
are selecting the most significant bits, and then $\mu$ can be
implemented as $y\mapsto y\rs (b-\ell)$. Alternatively, if $j_i=b-i$,
then we are selecting the least significant bits, and then $\mu$ can
be implemented as $y\mapsto y\andtt (2^\ell-1)=y\andtt (r-1)$.

We assume a $k$-universal hash function $h:[u]\to[p]$, e.g.,
the one from \Cref{alg:Mersenne}. To get hash values in $[r]$,
we use $\mu\circ h$. Since $\mu$ is deterministic,
the hash values of up to $k$ distinct keys remain
independent with $\mu\circ h$. The issue is that hash values from
$\mu\circ h$ are not quite uniform in $[r]$.

Recall that for any key $x$, we have $h(x)$ uniformly distributed in $[2^b-1]$.
This is the uniform distribution on $b$-bit strings except that we are
missing $p=2^b-1$. Now $p$ is the all \texttt{1}s, and
$\mu(p) = r-1$.
Therefore
\begin{align}
	\text{for $i < r-1$,}\quad
	\Pr[\mu(h(x))=i]
	 & =\lceil p/r\rceil/p
	=((p+1)/r)/p
	=(1+1/p)/r
	\label{eq:coll-ell<r-1}
	\\
	\text{and}\quad
	\Pr[\mu(h(x))=r-1]
	 & =\lfloor p/r\rfloor/p=((p+1-r)/r)/p
	=(1-(r-1)/p)/r.
	\label{eq:coll-ell=r-1}
\end{align}
Thus $\Pr[\mu(h(x))=i]\leq (1+1/p)/r$ for all $i\in[r]$.
This upper-bound only has a relative error of $1/p$ from the uniform $1/r$.

Combining \eqref{eq:coll-ell<r-1} and \eqref{eq:coll-ell=r-1} with
pairwise independence, for any distinct keys $x,y\in [u]$, we show that the
collision probability is bounded
\begin{align}
	\Pr[\mu(h(x))=\mu(h(y))]
	 & =(r-1)((1+1/p)/r)^2+((1-(r-1)r/p)/r)^2 \nonumber
	\\&=(1+(r-1)/p^2)/r
	.\label{eq:coll}
\end{align}
Thus the relative error $r/p^2$ is small as long as $p$ is large.

\paragraph{The problem with non-Mersenne Primes}

Suppose $c\neq 1$ and we want to select arbitrary bits like in the arguments above.
If we pick the least significant bits we get a generic upper bound of $(1+c/p)/r$, which is not too bad for small $c$.
Here there is no conceptual difference to our Mersenne results.

However take the opposite extreme where
we pick just the one most significant bit
and $c=2^{b-1}-1$ (so $p=2^{b-1}+1$, a Fermat prime).
That bit is $0$ with probability $1-1/p$ and 1 only with probability $p$ -- virtually a constant.
We might try to fix this by xoring the output with a random number from $[2^b]$ (or add $C\in[2^b]$ and take mod $2^b$), but that will only make the bits uniform, not actually dependent on the key.
Thus if we hash two keys, $x_1$ and $x_2$, mod $2^{b-1}+1$ and take the top bit from each one, \emph{they will nearly always be the same, independent of whether $x_1=x_2$}.

More generally,
say we pick the $\ell$ most significant bits
and $c\leq 2^{b-1}-2^{b-\ell}$, then $2^{b-\ell}$ elements from
$[p]$ map to $0$ while only $\max\{0,2^{b-\ell}-c\}$
map to the all \texttt{1}s.
More concretely, take $\ell=b/2$ and $c=2^{b/2}\approx\sqrt{p}$ (typical for generalized Mersenne primes) \emph{then the top $\ell$ bits hit the all \texttt{1}s with 0 probability}, while the all \texttt{0}s is twice as common as the remaining values.

\subsection{Two-for-one hash functions in second moment estimation}
In this section, we discuss how we can get several hash functions for
the price of one, and apply the idea to second moment estimation using
Count Sketches \cite{charikar04count-sketch}.

Suppose we had a $k$-universal hash function into $b$-bit strings.
We note that using standard programming languages such as C, we have
no simple and efficient method of computing such hash
functions when $k>2$. However, later we will argue that polynomial
hashing using a Mersenne prime $2^b-1$ delivers a better-than-expected
approximation.

Let $h:U\to [2^b]$ be $k$-universal. By definition this
means that if we have $j\leq k$ distinct keys $x_0,\ldots,x_{j - 1}$, then
$(h(x_0),\ldots,h(x_{j - 1}))$ is uniform in $[2^b]^j\equiv [2]^{bj}$,
so this means that \emph{all} the bits in $h(x_0),\ldots,h(x_{j - 1})$ are
independent and uniform. We can use this to split our $b$-bit hash
values into smaller segments, and sometimes use them as if
they were the output of universally computed hash functions.

We illustrate this idea below in the context of the second moment estimation.
For this purpose the ``split'' we will be considering is into the first bit and the remaining bits.
\footnote{Note there are other ways to construct this sketch, which only use one hash function, such as~\cite{thorup12kwise}.
The following should thus not be taken as ``the only way'' to achieve this result, but as an example of how the ``intuitive approach'' turns out to work when hashing with Mersenne primes.}

\subsubsection{Second moment estimation}\label{sec:count-sketch}
We now review the second moment estimation of streams based on Count Sketches \cite{charikar04count-sketch} (which are based on the
celebrated second moment AMS-estimator from \cite{alon96frequency}.)

The basic setup is as follows:
For keys in $[u]$ and integer values in $\Z$, we are given a stream of key/value $(x_0,\Delta_0),\ldots, (x_{n-1},\Delta_{n-1})\in [u]\times\Z$. The
total value of key $x\in[u]$ is
\[f_x=\sum_{i\in[n],x_i=x} \Delta_i.\]
We let $n\leq u$ be  the number of non-zero values
$f_x\neq 0$, $x\in [u]$. Often $n$ is much smaller than $u$.
We define the $m$th moment $F_m = \sum_{x\in [u]}f_y^m$. The goal here is to
estimate the second moment $F_2 = \sum_{x\in [u]}f_x^2=\|f\|^2_2$.

\begin{algorithm}[H]
	\caption{\label{alg:count-sketch} Count Sketch. Uses a
	vector/array $C$ of $r$ integers and two independent
	4-universal hash functions $i:[u]\to[r]$ and $s:[u]\to\{-1,1\}$.
	}
	\begin{algorithmic}
		\Procedure{Initialize}{}
		\State For $i\in[t]$, set $C[i]\gets 0$.
		\EndProcedure
		\Procedure{Process}{$x, \Delta$}
		\State $C[i(x)]\gets C[i(x)]+s(x) \Delta$.
		\EndProcedure
		\Procedure{Output}{}
		\State \Return $\sum_{i\in[t]} C[i]^2$.
		\EndProcedure
	\end{algorithmic}
\end{algorithm}
The standard analysis \cite{charikar04count-sketch} shows that
\begin{align}
	\E[X]   & = F_2 \label{eq:E-F2}                      \\
	\Var[X] & =2(F_2^2 - F_4)/r<2F_2^2/r \label{eq:V-F2}
\end{align}
We see that by choosing larger and larger r we can make X concentrate around $F_2=\|f\|^2_2$. Here
$X=\sum_{i\in[r]} C[i]^2=\|C\|^2_2$. Now $C$ is a randomized function
of $f$, and as $r$ grows, we get $\|C(f)\|^2_2\approx\|f\|^2_2$,
implying $\|C(f)\|_2\approx\|f\|_2$, that is, the Euclidean norm is
statistically preserved by the Count Sketch. However, the Count Sketch
is also a linear function, so Euclidean distances are statistically
preserved, that is, for any $f,g\in \Z^u$,
\[\|f-g\|_2\approx \|C(f-g)\|_2=\|C(f)-C(g)\|_2.\]
Thus, when we want to find close vectors, we can just work with the
much smaller Count Sketches.
The count sketch $C$ can also be used to estimate any single value $f_x$.
To do this, we use the unbiased estimator $X_x=s(x)C[i(x)]$.
This is yet another standard use of count sketch \cite{charikar04count-sketch}.
It requires direct access to both the sketch $C$ and the two hash functions $s$ and $i$.
To get concentration one takes the median of multiple such estimators.

\subsubsection{Two-for-one hash functions with b-bit hash values}
As the count sketch is described above,
it uses two independent 4-universal hash functions
$i:[u]\to[r]$ and $s:[u]\to\{-1,1\}$, but 4-universal hash functions
are generally slow to compute, so, aiming to save roughly a factor 2
in speed, a tempting idea is to compute them both using a single hash
function.

The analysis behind \eqref{eq:E-F2} and \eqref{eq:V-F2} does not quite
require $i:[u]\to[r]$ and $s:[u]\to\{-1,1\}$ to be independent.
It suffices that the hash values are uniform and that for any
given set of $j\leq 4$ distinct keys $x_0,\ldots,x_{j - 1}$, the $2j$ hash
values $i(x_0),\ldots,i(x_{j - 1}),s(x_0),\ldots,s(x_{j - 1})$ are independent.
A critical step in the analysis is that if
a value $A$ depends on the first $j-1$ values ($A=A(i(x_0),\ldots,i(x_{j - 1}),s(x_1),\ldots,s(x_{j - 1}))$), but doesn't depend
on $s(x_0)$, then
\begin{equation}\label{eq:E-0}
	\E[s(x_0) A] = 0 .
\end{equation}
This follows because $\E[s(x_0)]=0$ by uniformity of $s(x_0)$ and because $s(x_0)$ is independent of $A$.

Assuming that $r=2^\ell$ is a power of two, we can easily construct
$i:[u]\to[r]$ and $s:[u]\to\{-1,1\}$ using a single $4$-universal
hash function $h:[u]\to[2^b]$ where $b>\ell$. Recall that all the bits in
$h(x_0),\ldots,h(x_3)$ are independent. We can therefore use the
$\ell$ least significant bits of $h(x)$ for $i(x)$ and the most
significant bit of $h(x)$ for a bit $a(x)\in[2]$, and finally set
$s(x)=1-2a(x)$. It is then easy to show that if $h$ is $4$-universal
then $h$ satisfies \cref{eq:E-0}.
\begin{algorithm}[H]
	\caption{For key $x\in [u]$, compute $i(x)=i_x\in[2^\ell]$ and
	$s(x)=s_x\in\{-1,1\}$,\rule{5ex}{0ex}
	using $h:[u]\to [2^b]$ where $b>\ell$.}
	\label{alg:h-and-s}
	\begin{algorithmic}
		\State $h_x\gets h(x)$
		\Comment $h_x$ uses $b$ bits
		\State $i_x\gets h_x \andtt (2^\ell-1)$
		\Comment $i_x$ gets $\ell$ least significant bits of $h_x$
		\State $a_x\gets h_x\rs (b-1)$
		\Comment $a_x$ gets the most significant bit of $h_x$
		\State $s_x\gets 1-(a_x\ls1)$
		\Comment $a_x\in[2]$ is converted to a sign $s_x\in\{-1,1\}$
	\end{algorithmic}
\end{algorithm}
Note that Algorithm \ref{alg:h-and-s} is well defined as long as
$h$ returns a $b$-bit integer. However, \cref{eq:E-0} requires
that $h$ is $k$-universal into $[2^b]$, which in particular implies that
the hash values are uniform in $[2^b]$.

\subsubsection{Two-for-one hashing with  Mersenne primes}\label{sec:two-for-one}
Above we discussed how useful it would be with $k$-universal hashing
mapping uniformly into $b$-bit strings. The issue was that the lack of
efficient implementations with standard portable code if
$k>2$. However, when $2^b-1$ is a Mersenne prime $p\geq u$, then we do
have the efficient computation from Algorithm \ref{alg:Mersenne}
of a $k$-universal hash function $h:[u]\to[2^b-1]$. The hash values
are $b$-bit integers, and they are uniformly distributed, except that
we are missing the all \texttt{1}s value $p=2^b-1$. We want to
understand how this missing value affects us if we try to split the
hash values as in Algorithm \ref{alg:h-and-s}. Thus, we assume a
$k$-universal hash function $h:[u]\to[2^b-1]$ from which we construct
$i:[u]\to[2^\ell]$ and $s:[u]\to\{-1,1\}$ as
described in Algorithm \ref{alg:h-and-s}. As usual, we assume $2^\ell>1$.
Since $i_x$ and $s_x$ are
both obtained by selection of bits from $h_x$, we know from Section
\ref{sec:power-of-two} that each of them have close to uniform
distributions. However, we need a good replacement for \eqref{eq:E-0}
which besides uniformity, requires $i_x$ and $s_x$ to be independent,
and this is certainly not the case.

Before getting into the analysis, we argue that we really do get two
hash functions for the price of one. The point is that our efficient
computation in Algorithm \ref{alg:Mersenne} requires that we use a
Mersenne prime $2^b-1$ such that $u\leq 2^{b-1}$, and this is even if
our final target is to produce just a single bit for the sign function
$s:[u]\to\{-1,1\}$. We also know that $2^\ell<u$, for otherwise we
get perfect results implementing $i:[u]\to[2^\ell]$ as the identity
function (perfect because it is collision-free).  Thus we can assume
$\ell<b$, hence that $h$ provides enough bits for both $s$ and $i$.

We now consider the effect of the hash values from $h$ being uniform
in $[2^b-1]$ instead of in $[2^b]$. Suppose we want to compute the
expected value of an expression $B$ depending only on the independent
hash values $h(x_0),\ldots,h(x_{j - 1})$ of $j\leq k$ distinct keys
$x_0,\ldots,x_{j - 1}$.

Our generic idea is to play with the distribution of $h(x_0)$ while
leaving the distributions of the other independent hash values
$h(x_0)\ldots,h(x_{j - 1})$ unchanged, that is, they remain uniform in
$[2^b-1]$. We will consider having $h(x_0)$ uniformly distributed in
$[2^b]$, denoted $h(x_0) \sim \unif[2^b]$, but then we later have to
subtract the ``fake'' case where $h(x_0)=p=2^b-1$.  Making the
distribution of $h(x_0)$ explicit, we get
\begin{equation}\begin{split}
		\E_{h(x_0) \sim \unif[p]}[B]&=\sum_{y\in[p]}\E[B \mid h(x_0)=y]/p
		\\&=\sum_{y\in[2^b]}\E[B \mid h(x_0)=y]/p - \E[B \mid h(x_0)=p]/p
		\\ &=\E_{h(x_0) \sim \unif[2^b]}[B](p+1)/p - \E[B \mid h(x_0)=p]/p.\label{eq:play-with-dist}
	\end{split}\end{equation}
Let us now apply this idea our situation where $i:[u]\to[2^\ell]$ and
$s:[u]\to\{-1,1\}$ are constructed from $h$ as described in Algorithm
\ref{alg:h-and-s}. We will prove
\begin{lemma}\label{lem:remove-si}  Consider distinct keys $x_0,\ldots,x_{j - 1}$, $j\leq k$ and an expression $B=s(x_0)A$ where $A$
	depends on $i(x_0),\ldots,i(x_{j - 1})$ and $s(x_1),\ldots,s(x_{j - 1})$ but not
	$s(x_0)$. Then
	\begin{equation}\label{eq:remove-si}
		\E[s(x_0)A]=\E[A\mid i(x_0)=2^\ell-1]/p.
	\end{equation}
\end{lemma}
\begin{proof}
	When $h(x_0) \sim \unif[2^b]$, then $s(x_0)$ is uniform
	in $\{-1,1\}$ and independent of $i(x_0)$. The remaining
	$(i(x_i),s(x_i))$, $i\ge 1$, are independent of $s(x_0)$ because they
	are functions of $h(x_i)$ which is independent of $h(x_0)$, so
	we conclude that
	\[\E_{h(x) \sim \unif[2^b]}[s(x_0)A]=0\]
	Finally, when $h(x_0)=p$, we get $s(x_0)=-1$ and $i(x_0)=2^\ell-1$,
	so applying \eqref{eq:play-with-dist}, we conclude
	that
	\[\E[s(x_0)A] = -\E[s(x_0) A \mid h(x_0) = p]/p = \E[A \mid i(x_0)=2^\ell-1]/p.\]
\end{proof}
Above \eqref{eq:remove-si} is our replacement for \eqref{eq:E-0}, that is,
when the hash values from $h$ are uniform in $[2^b-1]$ instead of
in $[2^b]$, then $\E[s(x_0)B]$ is reduced by $\E[B \mid i(x_0)=2^\ell-1]/p$.
For large $p$, this is a small additive error. Using this in a careful
analysis, we will show that our fast second moment estimation
based on Mersenne primes performs almost perfectly:

\begin{theorem}\label{thm:h-and-s-p}
	Let $r>1$ and $u>r$ be powers of two and let $p=2^b-1>u$ be a
	Mersenne prime.
	Suppose we have a 4-universal hash function $h:[u]\to[2^b-1]$, e.g.,
	generated using Algorithm \ref{alg:Mersenne}. Suppose
	$i:[u]\to[r]$ and
	$s:[u]\to\{-1,1\}$ are constructed from $h$ as described in
	Algorithm \ref{alg:h-and-s}. Using this $i$ and $s$
	in the Count Sketch Algorithm \ref{alg:count-sketch}, the second moment
	estimate $X=\sum_{i\in[k]} C_i^2$ satisfies:
	\begin{align*}
		\E[X] < (1+n/p^2)\,F_2,
		\quad
		|\E[X] - F_2 | \le F_2 (n - 1)/p^2,
		\quad
		\Var[X]< 2F_2^2/r.
	\end{align*}
\end{theorem}
The difference from \eqref{eq:E-F2} and \eqref{eq:V-F2}
is negligible when $p$ is large. Theorem \ref{thm:h-and-s-p} will be
proved in Section \ref{sec:analysis-two-for-one}.

Recall our discussion from the end of Section
\ref{sec:power-of-two}. If we instead had used the $b$-bit prime
$p=2^{b-1}+1$, then the sign-bit $a_x$ would be extremely biased with
$\Pr[a_x=0]=1-1/p$ while $\Pr[a_x=1]=1/p$, leading to extremely poor
performance.

\subsection{An arbitrary number of buckets}\label{sec:most-uniform}
We now consider the general case where we want to hash into a set of
buckets $R$ whose size is not a power of two.  Suppose we have a
$2$-universal hash function $h:U\to Q$.  We will compose $h$ with a
map $\mu:Q\to R$, and use $\mu\circ h$ as a hash function from $U$ to
$R$.  Let $q=|Q|$ and $r=|R|$.  We want the map $\mu$ to be \emph{most
	uniform} in the sense that for bucket $i\in R$, the number of
elements from $Q$ mapping to $i$ is either $\floor{q/r}$ or
$\ceil{q/r}$.  Then the uniformity of hash values with $h$ implies for
any key $x$ and bucket $i\in R$ \[\floor{q/r}/q\leq
	\Pr[\mu(h(x))=i]\leq \ceil{q/r}/q.\] Below we typically have $Q=[q]$
and $R=[r]$.  A standard example of a most uniform map $\mu:[q]\to[r]$
is $\mu(x)=x\bmod r$ which the one used above when we defined
$h^r:[u]\to[r]$, but as we mentioned before, the modulo operation is
quite slow unless $r$ is a power of two.

Another example of a most uniform map $\mu:[q]\to[r]$
is $\mu(x)=\floor{xr/q}$,
which is also quite slow in general, but if $q=2^b$ is a power of two,
it can be implemented as $\mu(x)=(xr)\rs\,b$ where
$\rs$ denotes right-shift. This would be yet another advantage
of having $k$-universal hashing into $[2^b]$.

Now, our interest is the case where $q$ is a Mersenne prime $p=2^b-1$. We want
an efficient and most uniform map $\mu:[2^b-1]$ into any given $[r]$.
Our simple solution is to define
\begin{equation}\label{eq:most-uniform}
	\mu(v)=\floor{(v+1)r/2^b}=((v+1)r)\rs b.
\end{equation}
Lemma \ref{lem:most-uniform} (iii) below
states that \eqref{eq:most-uniform} indeed
gives a most uniform map.
\begin{lemma}\label{lem:most-uniform} Let $r$ and $b$ be positive integers.
	Then
	\begin{itemize}
		\item[(i)] $v\mapsto (vr)\rs\,b$ is a most
		      uniform map from $[2^b]$ to $[r]$.
		\item[(ii)] $v\mapsto (vr)\rs\,b$ is a most
		      uniform map from $[2^b]\setminus\{0\}=\{1,\ldots,2^b-1\}$ to $[r]$.
		\item[(iii)] $v\mapsto ((v+1)r)\rs \, b$ is a most
		      uniform map from $[2^b-1]$ to $[r]$.
	\end{itemize}
\end{lemma}
\begin{proof}
	Trivially (ii) implies (iii).
	The statement (i) is folklore and easy to prove, so we know that every
	$i\in[r]$ gets hit by $\floor {2^b/r}$ or $\ceil{2^b/r}$ elements from
	$[2^b]$. It is also clear that $\ceil{2^b/r}$ elements, including $0$,
	map to $0$. To prove (ii), we remove $0$ from $[2^b]$,
	implying that only
	$\ceil{2^b/r}-1$ elements map to $0$. For all positive integers $q$
	and $r$, $\ceil{(q+1)/r}-1=\floor{q/r}$, and we use this here with
	$q=2^b-1$. It follows that all buckets from $[r]$ get $\floor{q/r}$
	or $\floor{q/r}+1$ elements from $Q=\{1,\ldots,q\}$. If $r$ does
	not divide $q$ then $\floor{q/r}+1=\ceil{q/r}$, as desired. However,
	if $r$ divides $q$, then $\floor{q/r}=q/r$, and this
	is the least number of elements from $Q$ hitting any bucket in $[r]$. Then
	no bucket from $[r]$ can get hit by more than $q/r=\ceil{q/r}$
	elements from $Q$. This completes the proof of (ii), and hence of (iii).
\end{proof}
We note that our trick does not work when $q=2^b-c$ for $c\geq 2$, that is,
using $v\mapsto ((v+c)r)\rs  b$, for in this general case,
the number of elements hashing to $0$ is $\ceil {2^b/r}-c$, or $0$ if
$c\geq \floor {2^b/r}$.
One may try many other hash functions $(c_1 v r+ c_2 v+ c_3 r + c_4) \rs b$ similarly without any luck.
Our new uniform map from \eqref{eq:most-uniform} is thus very specific to Mersenne prime fields.


\subsection{Division and Modulo with (Pseudo) Mersenne Primes}\label{subsec:intro-division}

We now describe a new algorithm for truncated division with Mersenne primes, and more generalized numbers on the form $2^b-c$.
We show this implies a fast branch-free computation of $\bmod\,p$ for
Mersenne primes $p=2^b-1$.
An annoyance in Algorithm \ref{alg:Mersenne}
is that the if-statement at the end can be slow in case of branch mis-predictions.
This method solves that issue.

More specifically, in Algorithm \ref{alg:Mersenne}, after the last
multiplication, we have a number $y<p^2$ and we want to compute the
final hash value $y\bmod p$. We obtained this using the following
statements, each of which preserves the value modulo $p$, starting from
$y<p^2$:
\begin{algorithmic}
	\State $y \gets (y\andtt p)+(y\rs b)$
	\Comment $y<2p$
	\If{$y\ge p$}
	\State $y\gets y-p$
	\Comment  $y<p$
	\EndIf
\end{algorithmic}
To avoid the if-statement, in Algorithm \ref{alg:div-simple}, we suggest
a branch-free code that starting
from $v<2^{2b}$ computes both $y=v\bmod p$ and $z=\floor{v/p}$ using
a small number of AC$^0$ instructions.
\begin{algorithm}[H]
	\caption{For Mersenne prime $p=2^b-1$ and $v< 2^{2b}$, compute
		\label{alg:div-simple}
		$y=v\bmod p$ and $z=\floor{v/p}$}
	\begin{algorithmic}
		\State $\rhd$ First we compute $z=\floor{v/p}$
		\State $v'=v+1$
		\State $z \gets(( v' \rs b)+v')\rs b$
		\State $\rhd$ Next we compute $y=v\bmod p$ given $z=\floor{v/p}$
		\State $y \gets (v + z) \andtt p $
	\end{algorithmic}
\end{algorithm}
In Algorithm \ref{alg:div-simple}, we use
$z=\floor{v/p}$ to compute $y=v\bmod p$. If we only want the
division $z=\floor{v/p}$, then we can skip the last statement.

Below we will generalize Algorithm \ref{alg:div-simple} to work for
arbitrary $v$, not only $v<2^{2b}$. Moreover, we will generalize
to work for different kinds of primes generalizing Mersenne primes:
\begin{description}
	\item[Pseudo-Mersenne Primes]
	      are primes of the form $2^b-c$, where is usually required that $c < 2^{\lfloor b/2\rfloor}$~\cite{van2014encyclopedia}.
	      Crandal patented a method for working with Pseudo-Mersenne Primes in 1992~\cite{crandall1992method},
	      why those primes are also sometimes called ``Crandal-primes''.
	      The method was formalized and extended by Jaewook Chung and Anwar Hasan in 2003~\cite{chung2003more}. The method we present is simpler with
	      stronger guarantees and better practical performance.
	      We provide a comparison with the Crandal-Chung-Hansan method in Section 4.
	\item[Generalized Mersenne Primes]
	      also sometimes known as Solinas primes~\cite{Solinas2011}, are sparse numbers, that is $f(2^b)$ where $f(x)$ is a low-degree polynomial.
	      Examples from the Internet Research Task Force's document ``Elliptic Curves for Security''~\cite{rfc7748}:
	      $p_{25519} = 2^{255} - 19$
	      and
	      $p_{448} = 2^{448}-2^{224}-1$.
	      We simply note that Solinas primes form a special case of
	      Pseudo-Mersenne Primes, where multiplication with $c$
	      can be done using a few shifts and additions.
\end{description}
We will now first generalize the division from Algorithm \ref{alg:div-simple} to cover arbitrary $v$ and division with an arbitrary Pseudo-Mersenne primes $p=2^b-c$.
This is done in Algorithm \ref{alg:division-generalized} below which
works also if $p=2^b-c$ is not a prime.  The
simple division in Algorithm \ref{alg:div-simple} corresponds to the case
where $c=1$ and $m=2$.
\begin{algorithm}[H]
	\caption{Given integers $p=2^b-c$ and $m$.
		For any $v< (2^b/c)^m$, compute $z=\floor{v/p}$}
	\label{alg:division-generalized}
	\begin{algorithmic}
		\State $v' \gets v + c$
		\State $z \gets v' \rs b$
		\For{ $m-1$ times}
		\State $z \gets (z * c + v')\rs b$
		\EndFor
	\end{algorithmic}
\end{algorithm}
The proof that Algorithm \ref{alg:division-generalized} correctly computes
$z=\floor{v/p}$ is provided in Section \ref{sec:division}.
Note that $m$ can be computed in advance from $p$, and there is no requirement that it is chosen as small as possible.
For Mersenne and Solinas primes, the multiplication $z*c$ can be done very fast.

Mathematically the algorithm computes the nested division
$$
	\bbfloor{\frac{v}{q-c}}
	=
	\bbfloor{\frac{
			\bfloor{\frac{
					\floor{\frac{
							\dots+v+c
						}{q}}c +v+c
				}{q}}c +v+c
		}{q}}
	\vspace{-1em} 
$$
which is visually similar to the series expansion
$
	\frac{v}{q-c}
	= \frac{v}{q}\sum_{i=0}^\infty (\frac{c}{q})^i
	= \frac{\frac{\frac{\dots+v}{q}c+v}{q}c+v}{q}.
$
It is natural to truncate this after $m$ steps for a $(c/q)^m$ approximation.
The less intuitive part is that we need to add $v+c$ rather than $v$ at each step, to compensate for rounding down the intermediate divisions.

\paragraph{Computing mod}
We will now compute the $\bmod$ operation assuming that
we have already computed $z=\floor{v/p}$. Then
\begin{align}
	v \bmod p
	= v - pz
	= v - (2^b-c)z
	= v - (z\ls b) - c*z,
\end{align}
which is only two additions, a shift, and a multiplication with $c$ on top of the division algorithm.
As $pz = \floor{v/p}p \le v$ there is no danger of overflow.
We can save one operation by noting
that if $v = z (2^b-c) + y$, then
$$v\bmod p = y=\left(v+c*z \right) \bmod 2^b.$$
This is the method presented in Algorithm \ref{alg:mod-generalized} and applied with $c=1$ in Algorithm \ref{alg:div-simple}.
\begin{algorithm}[H]
	\caption{For integers $p=2^b-c$ and $z=\floor{v/p}$ compute
		$y=v \bmod p$.}
	\label{alg:mod-generalized}
	\begin{algorithmic}
		\State $y \gets (v + z*c) \andtt (2^b-1)$
	\end{algorithmic}
\end{algorithm}


 \paragraph{Applications to an arbitrary number of buckets}
 In Subsection~\ref{sec:most-uniform} we discussed how $\floor{\frac{h(x)r}{2^b-1}}$ provides a most uniform map from $[2^b-1]\to[r]$.
 To avoid the division step, we instead considered the map
 $\floor{\frac{(h(x)+1)r}{2^b}}$.
 However, for primes of the form $2^b-c$, $c>1$ this approach doesn't provide a most-uniform map.
 Instead, we may use Algorithm \ref{alg:division-generalized} to compute
 $$\left\lfloor\frac{h(x)r}{2^b-c}\right\rfloor$$
 directly, getting a perfect most-uniform map.
 
 \paragraph{Application to Finger Printing}
 A classical idea by Rabin~\cite{rabin1981fingerprinting} is to test the equality of two large numbers by comparing their value modulo some random primes (or random irreducible polynomials in a Gallois Field.)
 A beautiful example of this is King and Sagert's Algorithm for Maintaining the Transitive Closure.~\cite{DBLP:journals/jcss/KingS02}
 For such applications
 we need a reasonably large set of random primes to choose from.
 The generalized Mersenne primes with $c$ up to $2^{b(1-\eps)}$ are a good candidate set, which from the prime number theorem we expect to contain $\approx 2^{b(1-\eps)}/b$ primes.
 Each application of \cref{alg:division-generalized} reduces $y$ by a factor $2^{-\eps b}$,
 so computing the quotient and remainder takes just $m=1/\eps$ steps.




\section{Analysis of second moment estimation using Mersenne primes}
\label{sec:analysis-two-for-one}
In this section, we will prove Theorem \ref{thm:h-and-s-p}---that a single Mersenne hash function works for Count Sketch.
Recall that for each key $x\in [u]$, we have a value $f_x\in \Z$, and the
goal was to estimate the second moment $F_2 = \sum_{x\in u}f_x^2$.

We had two functions $i:[u]\to[r]$ and $s:[u]\to\{-1,1\}$. 
For notational convenience, we define $i_x=i(x)$ and $s_x=s(x)$.
We let $r=2^\ell>1$ and $u>r$ both be powers of two and $p=2^b-1>u$ a Mersenne prime.
For each $i\in [r]$, we have a counter 
$C_i=\sum_{x\in[u]} s_x f_x[i_x=i]$, and we define the 
estimator $X=\sum_{i\in[r]} C_i^2$. We want to study how
well it approximates $F_2$.
We have 
\begin{align}
X=\sum_{i\in[r]}\left( \sum_{x\in[u]}s_x f_x[i_x=i]\right)^2
=\sum_{x,y\in[u]}s_x s_y f_x f_y[i_x=i_y]
=\sum_{x\in[u]} f_x^2+Y,
\label{eq:decomp}
\end{align}
where $Y=\sum_{x,y\in[u],x\neq y} s_x s_y f_x f_y [i_x = i_y]$.
The goal is thus to bound mean and variance of the error $Y$.

As discussed in the introduction, one of the critical steps in the analysis of count sketch in the classical case is \cref{eq:E-0}.
We formalize this into the following property:
\begin{property}[Sign Cancellation]\label{prop:independence}
    For distinct keys $x_0, \ldots x_{j - 1}$, $j \le k$
    and an expression $A(i_{x_0}, \ldots, i_{x_{j - 1}}, s_{x_1}, \ldots, s_{x_{j - 1}})$,
    which depends on $i_{x_0}, \ldots, i_{x_{j - 1}}$ and $s_{x_1}, \ldots, s_{x_{j - 1}}$
    but not on $s_{x_0}$
    \begin{align}
        \E[s_{x_0} A(i_{x_0}, \ldots, i_{x_{j - 1}}, s_{x_1}, \ldots, s_{x_{j - 1}})] = 0\; .
    \end{align}
\end{property}

In the case where we use a Mersenne prime for our hash function we have that $h$ is uniform in $[2^b - 1]$ and not in $[2^b]$, hence \Cref{prop:independence} is not satisfied.
Instead, we have \cref{eq:E-0} which is almost as good, and will replace \Cref{prop:independence} in the analysis for count sketch.
We formalize this as follows:
\begin{property}[Sign Near Cancellation]\label{prop:near-independence}
    Given $k, p$ and $\delta$,
    there exists $t \in [r]$ such that for distinct keys $x_0, \ldots x_{j - 1}$, $j \le k$
    and an expression $A(i_{x_0}, \ldots, i_{x_{j - 1}}, s_{x_1}, \dots, s_{x_{j - 1}})$,
    which depends on $i_{x_0}, \ldots, i_{x_{j - 1}}$
    and $s_{x_1}, \ldots, s_{x_{j - 1}}$,
    but not on $s_{x_0}$,
    \begin{align}
        \E[s_{x_0} A(i_{x_0}, \ldots, i_{x_{j - 1}}, s_{x_1}, \ldots, s_{x_{j - 1}})]
            &= \frac1p \E[A(i_{x_0}, \ldots, i_{x_{j - 1}}, s_{x_1}, \ldots, s_{x_{j - 1}}) \mid i_{x_0} = t].
         \label{eq:near-independence}
            \\
            \text{and}\quad
    \Pr[i_x = t] &\le (1 + \delta)/r
    \quad\text{for any key $x$}.
         \label{eq:prob-special-value}
    \end{align}
\end{property}

When the hash function $h$ is not uniform then it is not guaranteed that
the collision probability is $1/r$, but \eqref{eq:coll} showed that for
Mersenne primes the collision probability is $(1 + (r - 1)/p^2)/r$.
We formalize this into the following property.
\begin{property}[Low Collisions]\label{prop:collision}
   We say the hash function has $(1+\eps)/r$-low collision probability, if
    for distinct keys $x \neq y$,
    \begin{align}\label{eq:collision}
        \Pr[i_x = i_y] \le (1 + \eps)/r\; .
    \end{align}
\end{property}

\subsection{The analysis in the classical case}
First, as a warm-up for later comparison, we analyse the
case where we have Sign Cancellation, but
the collision probability bound is only $(1+\eps)/r$.
This will come in useful in \Cref{sec:arbitrary-buckets} where we will consider the case of an arbitrary number of buckets, not necessarily a power of two.
\begin{lemma}\label{lem:count-classic}
   If the hash function has Sign Cancellation for $k = 4$ and $(1+\eps)/r$-low collision probability, then
    \begin{align}
        \E[X] &= F_2 \\
        \Var[X] &\le 2(1 + \eps)(F_2^2 - F_4)/r \le 2(1 + \eps)F_2^2/r .
    \end{align}
\end{lemma}
\begin{proof}
   Recall the decomposition $X=F_2+Y$ from \cref{eq:decomp}.
    We will first show that $\E[Y] = 0$.
    By \Cref{prop:independence} we have that $\E[s_x s_y f_x f_y [i_x = i_y]] = 0$
    for $x \neq y$ and thus $\E[Y] = \sum_{x,y\in[u],x\neq y} \E[s_x s_y f_x f_y [i_x = i_y]] = 0$.

    Now we want to bound the variance of $X$. We note that since $\E[Y] = 0$ and $X = F_2 + Y$
    \begin{align*}
        \Var[X] = \Var[Y] = \E[Y^2]
            = \sum_{\substack{x, y, x', y' \in [u]\\ x \neq y, x' \neq y'}} \E[(s_x s_y f_x f_y [i_x = i_y])(s_{x'} s_{y'} f_{x'} f_{y'} [i_{x'} = i_{y'}])] .
    \end{align*}
    Now we consider one of the terms $\E[(s_x s_y f_x f_y [i_x = i_y])(s_{x'} s_{y'} f_{x'} f_{y'} [i_{x'} = i_{y'}])]$.
    Suppose that one of the keys, say $x$, is unique, i.e. $x \not\in \{y, x', y'\}$.
    Then the Sign Cancellation Property implies that 
    \[
        \E[(s_x s_y f_x f_y [i_x = i_y])(s_{x'} s_{y'} f_{x'} f_{y'} [i_{x'} = i_{y'}])] = 0 .
    \]
    Thus we can now assume that there are no unique keys. Since $x \neq y$ and $x' \neq y'$, we conclude
    that $(x, y) = (x', y')$ or $(x, y) = (y', x')$. Therefore
    \begin{align*}
       \Var[X] &= \sum_{\substack{x, y, x', y' \in [u]\\ x \neq y, x' \neq y'}}
                \E[(s_x s_y f_x f_y [i_x = i_y])(s_{x'} s_{y'} f_{x'} f_{y'} [i_{x'} = i_{y'}])]
            \\&= 2 \sum_{\substack{x, y, x', y' \in [u]\\ x \neq y, (x', y') = (x, y)}}
                \E[(s_x s_y f_x f_y [i_x = i_y])(s_{x'} s_{y'} f_{x'} f_{y'} [i_{x'} = i_{y'}])]
            \\&= 2\sum_{x,y\in[u],x\neq y} \E[(s_x s_y f_x f_y[i_x=i_y])^2]
            \\&= 2\sum_{x,y\in[u],x\neq y} \E[(f_x^2f_y^2[i_x=i_y])]
            \\&\le 2\sum_{x,y\in[u],x\neq y} (f_x^2f_y^2)(1 + \eps)/r
            \\&= 2(1 + \eps) (F_2^2-F_4)/r.
    \end{align*}
    The inequality follows by \Cref{prop:collision}.
\end{proof}



\subsection{The analysis of two-for-one using Mersenne primes}
We will now analyse the case where the functions $s : [u] \to \{-1, 1\}$
and $i : [u] \to [2^l]$ are constructed as in \Cref{alg:Mersenne} from a
single $k$-universal hash function $h : [u] \to [2^b - 1]$ where $2^b - 1$
is a Mersenne prime.
We now only have \nameref{prop:near-independence}.
We will show that this does
not change the expectation and variance too much. Similarly, to the
analysis of the classical case, we will analyse a slightly more general
problem, which will be useful in \Cref{sec:arbitrary-buckets}.
\begin{lemma}\label{lem:count-mersenne}
   If we have \nameref{prop:near-independence} with $
    \Pr[i_x = t] \le (1 + \delta)/r$
   and $(1+\eps)/r$-low collision probability,
   then
    \begin{align}
        \E[X] &= F_2 + (F_1^2 - F_2)/p^2 \\
        | \E[X] - F_2 | &\le F_2 (n - 1)/p^2 \\
        \Var[X] &\le 2F_2^2/r + F_2^2 (2\eps/r + 4(1 + \delta)n / (rp^2) + n^2/p^4 - 2 /(rn))
    \end{align}
\end{lemma}
\begin{proof}
    We first bound $\E[s_x s_y f_x f_y [i_x = i_y]]$ for distinct keys
    $x \neq y$.
    Let $t$ be the special index given by Sign Near Independence.
    Using \cref{eq:near-independence} twice we get that
    \begin{equation}\begin{split}\label{eq:twice-split}
        \E[s_x s_y f_x f_y [i_x = i_y]]
            &= \E[s_x f_x f_y [i_x = i_y] \mid i_y = t]/p
            \\&= \E[s_x f_x f_y [i_x = t]]/p
            \\&= \E[f_x f_y [i_x = t] \mid i_x = t]/p
            \\&= f_x f_y / p^2 \; .
    \end{split}\end{equation}
    From this, we can calculate $\E[X]$.
    \begin{align*}
        \E[X]
            = F_2 + \sum_{x \neq y} \E[s_x s_y f_x f_y [i_x = i_y]]
            = F_2 + (F_1^2 - F_2)/p^2 .
    \end{align*}
    Now we note that $0 \le F_1^2 \le n F_2$ by Cauchy-Schwarz, hence we get that
    $| \E[X] - F_2 | \le (n - 1)/p^2$.

    The same method is applied to the analysis of the variance, which is
    \[
        \Var[X]
            = \Var[Y]
            \le \E[Y^2]
            = \sum_{x,y,x',y' \in [u], x \neq y, x' \neq y'} \E[(s_x s_y f_x f_y [i_x = i_y]) (s_{x'} s_{y'} f_{x'} f_{y'}[i_{x'} = i_{y'}])]
        \; .
    \] 
    Consider any term in the sum. Suppose some key, say $x$, is unique in the
    sense that $x \not \in \{y,x',y'\}$. Then we can apply \cref{eq:near-independence}.
    Given that $x \neq y$ and $x'\neq y'$, we have either $2$ or $4$ such unique keys.
    If all 4 keys are distinct, as in \cref{eq:twice-split}, we get
    \begin{align*}
        \E[(s_x s_y f_x f_y [i_x = i_y]) &(s_{x'} s_{y'} f_{x'} f_{y'}[i_{x'} = i_{y'}])]
            \\&= \E[(s_x s_y f_x f_y [i_x = i_y])] \E[s_{x'} s_{y'} f_{x'} f_{y'}[i_{x'} = i_{y'}])]
            \\&= (f_x f_y/p^2)(f_{x'} f_{y'}/p^2)
            \\&= f_x f_y f_{x'} f_{y'}/p^4
        \; .
    \end{align*}
    The expected sum over all such terms is thus bounded
    as 
    \begin{equation}\begin{split}
        \sum_{{\rm distinct}\, x,y,x',y'\in[u]}& \E[(s_x s_y f_x f_y [i_x = i_y]) (s_{x'} s_{y'} f_{x'} f_{y'}[i_{x'} = i_{y'}])]
            \\&= \sum_{{\rm distinct}\,x,y,x',y'\in[u]} f_xf_yf_{x'}f_{y'}/p^4
            \\&\le F_1^4 /p^4
            \\&\le F_2^2 n^2/p^4.\label{eq:distinct}
    \end{split}\end{equation}
    Where the last inequality used Cauchy-Schwarz. We also have to consider all the cases with
    two unique keys, e.g., $x$ and $x'$ unique while $y=y'$. Then using \cref{eq:near-independence}
    and \cref{eq:prob-special-value}, we get
    \begin{align*}
        \E[(s_x s_y f_x f_y [i_x = i_y]) &(s_{x'} s_{y'} f_{x'} f_{y'}[i_{x'} = i_{y'}])]
            \\&= f_x f_{x'} f_y^2 \E[s_x s_{x'} [i_x = i_{x'} = i_y]]
            \\&= f_x f_{x'} f_y^2 \E[s_{x'} [t = i_{x'} = i_y]]/p
            \\&= f_x f_{x'} f_y^2 \E[t = i_y]/p^2
            \\&\le f_x f_{x'} f_y^2(1 + \delta)/(rp^2).
    \end{align*}    
    Summing over all terms with $x$ and $x'$ unique while $y=y'$, and
    using Cauchy-Schwarz and $u\leq p$, we get 
    \begin{align*}
        \sum_{{\rm distinct}\,x,x',y} f_xf_{x'}f_y^2 (1 + \delta) /(rp^2) 
            \le F_1^2 F_2 (1 + \delta)/(rp^2)
            \le F_2^2 n(1 + \delta)/(rp^2).
    \end{align*}
    There are four ways we can pick the two unique keys $a\in \{x,y\}$
    and $b\in \{x',y'\}$, so we conclude that
    \begin{equation}\label{eq:one-pair}
        \sum_{\substack{
            x,y,x',y'\in[u], x\neq y, x'\neq y',\\
            (x,y)=(x',y')\,\vee\,(x,y)=(y',x')
        }}
        \E[(s_x s_y f_x f_y [i_x = i_y]) (s_{x'} s_{y'} f_{x'} f_{y'}[i_{x'} = i_{y'}])]
            \le 4 F_2^2 n(1 + \delta)/(rp^2) .
    \end{equation}
    Finally, we need to reconsider the terms with two pairs, that
    is where $(x,y)=(x',y')$ or $(x,y)=(y',x')$. In
    this case, $(s_x s_y f_x f_y [i_x = i_y]) (s_{x'} s_{y'} f_{x'} f_{y'}[i_{x'} = i_{y'}]) = f_x^2 f_y^2 [i_x = i_y]$.
    By \cref{eq:collision}, we get 
    \begin{equation}\begin{split}    
        \sum_{\substack{
            x,y,x',y'\in[u], x\neq y, x'\neq y',\\
            (x,y)=(x',y')\,\vee\,(x,y)=(y',x')
        }}&
            \E[(s_x s_y f_x f_y [i_x = i_y]) (s_{x'} s_{y'} f_{x'} f_{y'}[i_{x'} = i_{y'}])]
            \\&=2\sum_{x,y\in[u],x\neq y} f_x^2f_y^2 \Pr[i_x=i_y]
            \\&=2\sum_{x,y\in[u],x\neq y} f_x^2f_y^2 (1 + \eps)/r
            \\&=2(F_2^2 - F_4)(1 + \eps)/r .\label{eq:two-pairs}
    \end{split}\end{equation}
    Adding up add \eqref{eq:distinct}, \eqref{eq:one-pair}, and
    \eqref{eq:two-pairs}, we get 
    \begin{align*}
        \Var[Y]
            &\le 2(1 + \eps)(F_2^2 - F_4)/r + F_2^2(4(1 + \delta) n / (rp^2) + n^2/p^4)
            \\&\le 2F_2^2/r + F_2^2 (2\eps/r + 4(1 + \delta)n / (rp^2) + n^2/p^4 - 2 /(rn)) .
    \end{align*}
    This finishes the proof.
\end{proof}

We are now ready to prove \Cref{thm:h-and-s-p}.
\begingroup
    \def\thelemma{\ref{thm:h-and-s-p}}
    \begin{theorem}
        Let $r>1$ and $u>r$ be powers of two and let $p=2^b-1>u$ be a
        Mersenne prime.
        Suppose we have a 4-universal hash function $h:[u]\to[2^b-1]$, e.g.,
        generated using Algorithm \ref{alg:Mersenne}. Suppose
        $i:[u]\to[r]$ and
        $s:[u]\to\{-1,1\}$ are constructed from $h$ as described in
        Algorithm \ref{alg:h-and-s}. Using this $i$ and $s$ 
        in the Count Sketch Algorithm \ref{alg:count-sketch}, the second moment 
        estimate $X=\sum_{i\in[k]} C_i^2$ satisfies:
        \begin{align}
           \E[X] &= F_2+(F_1^2-F_2)/p^2, \label{eq:E-F2-p}\\
           | \E[X] - F_2 | &\le F_2 (n - 1)/p^2, \label{eq:E-F2-p-com}\\
           \Var[X]&< 2F_2^2/r.\label{eq:V-F2-p}
        \end{align}
    \end{theorem}
    \addtocounter{lemma}{-1}
\endgroup

From \Cref{eq:remove-si} and \Cref{eq:coll-ell=r-1} we have
\nameref{prop:near-independence} with $
    \Pr[i_x = 2^b - 1] \le (1 - (r - 1)/p)/r$
   and \Cref{eq:coll} $(1 + (r - 1)/p^2)/r$-low collision probability.
Now \Cref{lem:count-mersenne} give us \eqref{eq:E-F2-p}
and \eqref{eq:E-F2-p-com}. Furthermore, we have that
\begin{align*}
    \Var[X] 
        &\le 2F_2^2/r + F_2^2 (2\eps/r + 4(1 + \delta)n / (rp^2) + n^2/p^4 - 2 /(rn))
        \\&= 2F_2^2/r + F_2^2(2/p^2 + 4n/(r p^2) + n^2/p^4 - 2/(rn)) .
\end{align*}

We know that $2 \le r \le u/2 \le (p + 1)/4$ and $n \le u$.
This implies that $p \ge 7$ and that $n/p \le u/p \le 4/7$.
We want to prove that
$2/p^2 + 4n/(r p^2) + n^2/p^4 - 2/(rn) \le 0$ which would
prove our result. We get that
\begin{align*}
    2/p^2 + 4n/(r p^2) + n^2/p^4 - 2/(rn)
        \le 2/p^2 + 4u/(r p^2) + u^2/p^4 - 2/(ru) .
\end{align*}
Now we note that $4u/(r p^2) - 2/(ru) = (2u^2 - p^2)/(u p^2 r) \le 0$
since $u \le (p + 1)/2$ so it maximized when $r = u/2$. We then get
that
\begin{align*}
    2/p^2 + 4u/(r p^2) + u^2/p^4 - 2/(ru)
        \le 2/p^2 + 8/p^2 + u^2 / p^4 - 4/u^2 .
\end{align*}
We now use that $u/p \le (4/7)^2$ and get that
\begin{align*}
    2/p^2 + 8/p^2 + u^2 / p^4 - 4/u^2
        \le (10 + (4/7)^2 - 4 (7/4)^2)/p^2
        \le 0 .
\end{align*}
This finishes the proof of \eqref{eq:V-F2-p} and thus also of \Cref{thm:h-and-s-p}.

\section{Algorithms and analysis with an arbitrary number of buckets}\label{sec:arbitrary-buckets}

Picture a hash function as throwing balls (keys) into buckets (hash values).
In the previous sections we have considered the case of a prime number of buckets.
We now consider the case of an arbitrary, not necessarily prime, number of buckets.

We will analyse the collision probability with the most uniform maps introduced in Section \ref{sec:most-uniform}, and later we will show how they can be used in connection with the two-for-one hashing from Section \ref{sec:two-for-one}.

\subsection{An arbitrary number of buckets}

We have a hash function $h:U\to Q$, but we want hash values in $R$, so we need a map $\mu:Q\to R$, and then use $\mu\circ h$ as our hash function from $U$ to $R$.
We normally assume that the hash values with $h$ are pairwise independent, that is, for any distinct $x$ and $y$, the hash values $h(x)$ and $h(y)$ are independent, but then $\mu(h(x))$ and $\mu(h(y))$ are also independent.
This means that the collision probability can be calculated as \[\Pr[\mu(h(x))=\mu(h(y))]=\sum_{i\in R}\Pr[\mu(h(x))=\mu(h(y))=i]=\sum_{i\in R}\Pr[\mu(h(x)=i)]^2.
\]
This sum of squared probabilities attains is minimum value $1/|R|$
exactly when $\mu(h(x))$ is uniform in $R$.

Let $q=|Q|$ and $r=|R|$.
Suppose that $h$ is $2$-universal.
Then $h(x)$ is uniform in $Q$, and then we get the lowest collision probability with $\mu\circ h$ if $\mu$ is most uniform as defined in Section \ref{sec:most-uniform}, that is, the number of elements from $Q$ mapping to any $i\in[r]$ is either $\floor{q/r}$ or $\ceil{q/r}$.
To calculate the collision probability, let $a\in[r]$ be such that $r$ divides $q+a$.
Then the map $\mu$ maps $\ceil{q/r}=(q+a)/r$ balls to $r-a$ buckets and $\floor{q/r}=(q+a-r)/r$ balls to $a$ buckets.
For a key $x\in [u]$, we thus have $r-a$ buckets hit with probability $(1+a/q)/r$ and $a$ buckets hit with probability $(1-(r-a)/q)/r$.
The collision probability is then
\begin{equation}
   \begin{split}
      \Pr[\mu(h(x))=\mu(h(y))] &= (r-a)((1+a/q)/r)^2+a((1-(r-a)r/q)/r)^2   \\&=(1+a(r-a)/q^2)/r \\&\le \left(1+(r/(2q))^2\right)/r.
      \label{eq:coll-a}
   \end{split}
\end{equation}
Note that the above calculation generalizes the one for \eqref{eq:coll} which had $a=1$.
We will think of $(r/(2q))^2$ as the general relative rounding cost when we do not have any information about how $r$ divides $q$.

\subsection{Two-for-one hashing from uniform bits to an arbitrary number of buckets}
We will now briefly discuss how we get the two-for-one hash functions in count sketches with an arbitrary number $r$ of buckets based on a single $4$-universal hash function $h:[u]\to [2^b]$.
We want to construct the two hash functions $s:[u]\to\{-1,1\}$ and $i:[u]\to[r]$.
As usual the results with uniform $b$-bit strings will set the bar that we later compare with when from $h$ we get hash values that are only uniform in $[2^b-1]$.

The construction of $s$ and $i$ is presented in Algorithm \ref{alg:b-bit-arb-r}.
\begin{algorithm}[H]
   \caption{For key $x\in [u]$, compute $i(x)=i_x\in[r]$ and $s(x)=s_x\in\{-1,1\}$.
   \newline
   Uses 4-universal $h:[u]\to [2^b]$.}
   \label{alg:b-bit-arb-r}
   \begin{algorithmic}
      \State $h_x\gets h(x)$
      \Comment $h_x$ has $b$ uniform bits
      \State $j_x\gets h_x\andtt(2^{b-1}-1)$
      \Comment $j_x$ gets $b-1$ least significant bits of $h_x$
      \State $i_x\gets (r*j_x)\rs (b-1)$
      \Comment $i_x$ is most uniform in $[r]$
      \State $a_x\gets h_x\rs (b-1)$
      \Comment $a_x$ gets the most significant bit of $h_x$
      \State $s_x\gets (a_x\ls 1)-1$
      \Comment $s_x$ is uniform in $\{-1,1\}$ and independent of $i_x$.
   \end{algorithmic}
\end{algorithm}
The difference relative to Algorithm \ref{alg:h-and-s} is the computation of $i_x$ where we now first pick out the $(b-1)$-bit string $j_x$ from $h_x$, and then apply the most uniform map $(rj_x)\rs (b-1)$ to get $i_x$.
This does not affect $s_x$ which remains independent of $i_x$, hence we still have \nameref{prop:independence}.
But $i_x$ is no longer uniform in $[r]$ and only most uniform so by \eqref{eq:coll-a} we have $(1 + (r/2^b)^2)/r$-low collision probability.
Now \Cref{lem:count-classic} give us $\E[X] = F_2$ and
\begin{equation}
   \label{eq:Var-b-bit-arb-r} \Var[X] \le 2(F_2^2 - F_4)\left(1+(r/2^b)^2\right)/r \le 2 F_2^2\left(1 + (r/2^b)^2 \right)/r .
\end{equation}

\subsection{Two-for-one hashing from Mersenne primes to an arbitrary number of buckets}
We will now show how to get the two-for-one hash functions in count sketches with an arbitrary number $r$ of buckets based on a single $4$-universal hash function $h:[u]\to [2^b-1]$.
Again we want to construct the two hash functions $s:[u]\to\{-1,1\}$ and $i:[u]\to[r]$.
The construction will be the same as we had in Algorithm \ref{alg:b-bit-arb-r} when $h$ returned uniform values in $[2^b]$ with the change that we set $h_x\gets h(x)+1$, so that it becomes uniform in $[2^b]\setminus\{0\}$.
It is also convenient to swap the sign of the sign-bit $s_x$ setting $s_x\gets 2a_x - 1$ instead of $s_x\gets 1-2a_x$.
The basic reason is that this makes the analysis cleaner.
The resulting algorithm is presented as Algorithm \ref{alg:Mersenne-arb-r}.
\begin{algorithm}[H]
   \caption{For key $x\in [u]$, compute $i(x)=i_x\in[r]$ and
   $s(x)=s_x\in\{-1,1\}$.\rule{5ex}{0ex}
   Uses 4-universal $h:[u]\to [p]$ for Mersenne prime $p=2^b-1\geq u$.
   }
   \label{alg:Mersenne-arb-r}
   \begin{algorithmic}
      \State $h_x\gets h(x)+1$
      \Comment $h_x$ uses $b$ bits uniformly except $h_x\neq 0$
      \State $j_x\gets h_x\andtt(2^{b-1}-1)$
      \Comment $j_x$ gets $b-1$ least significant bits of $h_x$
      \State $i_x\gets (r*j_x)\rs (b-1)$
      \Comment $i_x$ is quite uniform in $[r]$
      \State $a_x\gets h_x\rs (b-1)$
      \Comment $a_x$ gets the most significant bit of $h_x$
      \State $s_x\gets 1-(a_x\ls1)$
      \Comment $s_x$ is quite uniform in $\{-1,1\}$ and quite independent of $i_x$.
   \end{algorithmic}
\end{algorithm}
The rest of Algorithm \ref{alg:Mersenne-arb-r} is exactly like Algorithm \ref{alg:b-bit-arb-r}, and we will now discuss the new distributions of the resulting variables.
We had $h_x$ uniform in $[2^b]\setminus\{0\}$, and then we set $j_x \gets h_x\andtt(2^{b-1}-1)$.
Then $j_x\in[2^{b-1}]$ with $\Pr[j_x=0]=1/(2^{b}-1)$ while  $\Pr[j_x=j]=2/(2^{b}-1)$ for all $j>0$.

Next we set $i_x\gets (rj_x)\rs b-1$.
We know from Lemma \ref{lem:most-uniform} (i) that this is a most uniform map from $[2^{b-1}]$ to $[r]$.
It maps a maximal number of elements from $[2^{b-1}]$ to $0$, including $0$ which had half probability for $j_x$.
We conclude
\begin{align}
   \Pr[i_x=0] & = (\ceil{2^{b-1}/r}2-1)/(2^{b}-1) \label{eq:prix0} \\ \Pr[i_x = i] & \in \{\floor{2^{b-1}/r}2/(2^{b}-1), \ceil{2^{b-1}/r}2/(2^{b}-1)\} \mbox{ for $i \neq  0$} \label{eq:prixneq0} .
\end{align}
We note that the probability for $0$ is in the middle of the two other bounds and often this yields a more uniform distribution on $[r]$ than the most uniform distribution we could get from the uniform distribution on $[2^{b-1}]$.

With more careful calculations, we can get some nicer bounds that we shall later use.
\begin{lemma}
   \label{lem:ix-r-dist}
   For any distinct $x,y\in [u]$,
   \begin{align}
      \Pr[i_x=0] & \le(1+r/2^b)/r\label{eq:ix=0} \\ \Pr[i_x=i_y] & \leq \left(1+(r/2^b)^2\right)/r.
      \label{eq:ix=iy}
   \end{align}
\end{lemma}
\begin{proof}
   The proof of \eqref{eq:ix=0} is a simple calculation.
   Using \eqref{eq:prix0} and the fact $\ceil{2^{b-1}/r}\le(2^{b-1}+r-1)/r$ we have
   \begin{align*}
      \Pr[i_x=0] & \le (2(2^{b-1}+r-1)/r)-1)/(2^{b}-1) \\    & =\left(1+(r-1)/(2^b-1)\right)/r\\   & \le\left(1+r/2^b\right)/r.
   \end{align*}
   The last inequality follows because $r<u<2^b$.

   For \ref{eq:ix=iy}, let $q=2^{b-1}$ and $p=1/(2q-1)$.
   We define $a\ge 0$ to be the smallest integer, such that\footnote{Like Knuth we let our divisibility symbol lean leftward.} $r\setminus q+a$.
   In particular this means $\lceil q/r\rceil = (q+a)/r$ and $\lfloor q/r\rfloor = (q-r+a)/r$.

   We bound the sum $$ \Pr[i_x=i_y] = \sum_{k=0}^{r-1} \Pr[i_x = k]^2 $$ by splitting into three cases: 1) The case $i_x=0$, where $\Pr[i_x=0]=(2\ceil{q/r}-1)p$, 2) the $r-a-1$ indices $j$ where $\Pr[i_x=j]=2\ceil{q/r}p$, and 3) the $a$ indices $j$ st.
   $\Pr[i_x=j]=2\floor{q/r}p$.
   \begin{align*}
      \Pr[i_x=i_y]
       & =
      (2p\ceil{ q/r}-p)^2 + (r-a-1) (2p \lceil q/r\rceil)^2 + (r-a) (2p \lfloor q/r \rfloor)^2
      \\  & = ((4a+1)r+4(q+a)(q-a-1))p^2/r
      \\  & \le (1 + (r^2-r)/(2q-1)^2) / r.
   \end{align*}
   The last inequality comes from maximizing over $a$, which yields $a=(r-1)/2$.

   The result now follows from
   \begin{align}
      (r^2-r)/(2q-1)^2 \le (r-1/2)^2/(2q-1)^2 \le (r/(2q))^2,
   \end{align}
   which holds exactly when $r\le q$.

\end{proof}
Lemma \ref{lem:ix-r-dist} above is all we need to know about the marginal distribution of $i_x$.
However, we also need a replacement for Lemma \ref{lem:remove-si} for handling the sign-bit $s_x$.
\begin{lemma}
   \label{lem:remove-si-r-dist}
   Consider distinct keys $x_0,\ldots,x_{j - 1}$, $j\leq k$ and an expression $B=s_{x_0}A$ where $A$ depends on $i_{x_0},\ldots,i_{x_{j - 1}}$ and $s_{x_1},\ldots,s_{x_{j - 1}}$ but not $s_{x_0}$.
   Then
   \begin{equation}
      \label{eq:remove-si-r-dist} \E[s_xA]=\E[A \mid i_x=0]/p.
   \end{equation}
\end{lemma}
\begin{proof}
   The proof follows the same idea as that for Lemma \ref{lem:remove-si}.
   First we have \[\E[B]=\E_{h(x_0) \sim \unif([2^b]\setminus\{0\})}[B]=\E_{h(x_0) \sim \unif[2^b]}[B]2^b/p-\E[B \mid h(x_0)=0]/p.
   \]
   With $h(x_0)\sim \unif[2^b]$, the bit $a_{x_0}$ is uniform and
   independent of $j_{x_0}$, so $s_{x_1}\in\{-1,1\}$ is uniform and
   independent of $i_{x_0}$, and therefore
   \[\E_{h(x_0) \sim \unif[2^b]}[s_{x_0}
         A]=0.
   \]
   Moreover, $h(x_0)=0$ implies $j_x={x_0}$, $i_{x_0}=0$, $a_{x_0}=0$,
   and $s_{x_0}=-1$,
   so
   \[\E[s_{x_0}
         A]=-\E[s_{x_0}A \mid h(x_0)=0]/p=\E[A \mid i_{x_0}=0].
   \]
\end{proof}

From \Cref{lem:remove-si-r-dist} and \eqref{eq:ix=0} we have \nameref{prop:near-independence} with $\Pr[i_x = 0] \le (1 + r/2^b)/r$, and \eqref{eq:ix=iy} implies that we have $(1 + (r/2^b)^2)/r$-low collision probability.
We can then use \Cref{lem:count-mersenne} to prove the following result.
\begin{theorem}
   \label{thm:h-and-s-p-arb-r}
   Let $u$ be a power of two, $1 < r \le u/2$, and let $p=2^b-1>u$ be a Mersenne prime.
   Suppose we have a 4-universal hash function $h:[u]\to[2^b-1]$, e.g., generated using \Cref{alg:Mersenne}.
   Suppose $i:[u]\to[r]$ and $s:[u]\to\{-1,1\}$ are constructed from $h$ as described in \Cref{alg:Mersenne-arb-r}.
   Using this $i$ and $s$ in the Count Sketch Algorithm \ref{alg:count-sketch}, the second moment estimate $X=\sum_{i\in[k]} C_i^2$ satisfies:
   \begin{align}
      \E[X] & = F_2+(F_1^2-F_2)/p^2, \label{eq:E-F2-p-arb-r} \\ | \E[X] - F_2 | & \le F_2 (n - 1)/p^2, \label{eq:E-F2-p-com-arb-r}\\ \Var[X] & < 2(1 + (r/2^b)^2)F_2^2/r.
      \label{eq:V-F2-p-arb-r}
   \end{align}
\end{theorem}

Now \Cref{lem:count-mersenne} gives us \eqref{eq:E-F2-p-arb-r} and \eqref{eq:E-F2-p-com-arb-r}.
Furthermore, we have that
\begin{align*}
   \Var[X] & \le 2F_2^2/r + F_2^2 (2\eps/r + 4(1 + \delta)n / (rp^2) + n^2/p^4 - 2 /(rn)) \\  & = 2(1 + (r/2^b)^2)F_2^2/r + F_2^2(4(1 + r/2^b)n/(r p^2) + n^2/p^4 - 2/(rn)) \\  & \le 2(1 + (r/2^b)^2)F_2^2/r + F_2^2(4(1 + r/p)n/(r p^2) + n^2/p^4 - 2/(rn)) .
\end{align*}

We know that $2 \le r \le u/2 \le (p + 1)/4$ and $n \le u$.
This implies that $p \ge 7$ and that $n/p \le u/p \le 4/7$.
If we can prove that $4(1 + r/p)n / (rp^2) + n^2/p^4 - 2 / (rn) \le 0$ then we have the result.
We have that
\begin{align*}
   4(1 + r/p)n / (rp^2) + n^2/p^4 - 2 / (rn) & = 4 n/(rp^2) + 4 n /(p^3) + n^2/p^4 - 2/(rn) \\  & \le 4u/(rp^2) + 4u/(p^3) + u^2/p^4 - 2/(ru) .
\end{align*}
Again we note that $4u/(r p^2) - 2/(ru) = (2u^2 - p^2)/(u p^2 r) \le 0$ since $u \le (p + 1)/2$ so it maximized when $r = u/2$.
We then get that
\begin{align*}
   4u/(r p^2) + 4u/(p^3) + u^2/p^4 - 2/(ru) \le 8/p^2 + 4u/(p^3) + u^2/p^4 - 4/u^2 .
\end{align*}
We now use that $u/p \le (4/7)^2$ and get that
\begin{align*}
   8/p^2 + 4u/(p^3) + u^2/p^4 - 4/u^2 .
   \le (8 + 4 (4/7) + (4/7)^2 - 4 (7/4)^2)/p^2
   \le 0 .
\end{align*}
This finishes the proof of \eqref{eq:V-F2-p-arb-r} and thus also of \Cref{thm:h-and-s-p-arb-r}.



\section{Quotient and Remainder with Generalized Mersenne Primes}
\label{sec:division}

The purpose of this section is to prove the correctness of Algorithm \ref{alg:division-generalized}.
In particular we will prove the following equivalent mathematical statement:

\begin{theorem}\label{thm:simple-div}
   Given integers $q>c>0$, $n\ge 1$ and a value $x$ such that
   $$0\le x \le \begin{cases}
      c (q/c)^{n} - c &\quad\text{if } c\!\setminus\!q \quad\text{($c$ divides $q$)} \\
      (q/c)^{n-1}(q-c) &\quad\text{otherwise}
   \end{cases}.$$
   Define the sequence $(v_i)_{i\in[n+1]}$ by
   $
      v_0 = 0$ and
      $v_{i+1} = \left\lfloor\frac{(v_i+1)c+x}{q}\right\rfloor$.
   Then
   $$
      \left\lfloor\frac{x}{q-c}\right\rfloor = v_n.$$
\end{theorem}
We note that when $c<q-1$ a sufficient requirement is that $x< (q/c)^n$.
For $c=q-1$ we are computing $\floor{x/1}$ so we do not need to run the algorithm at all.

To be more specific, the error $E_i = \floor{\frac{x}{q-c}} - v_i$ at each step
is bounded by $0\le E_i\le u_{n-i}$,
where $u_i$ is a sequence defined by
$u_0=0$ and $u_{i+1} = \lfloor\frac{q}{c}u_i+1\rfloor$.
For example, this means that if we stop the algorithm after $n-1$ steps, the error will be at most $u_1=1$.
\begin{proof}
   Write $x = m(q-c)+h$ for non-negative integers $m$ and $h$ with $h<q-c$.
   Then we get
   \begin{align*}
      \left\lfloor\frac{x}{q-c}\right\rfloor = m.
      \label{eq:floor}
   \end{align*}

   Let $u_0=0$, $u_{i+1} = \lfloor\frac{q}{c}u_i+1\rfloor$.
   By induction $u_i \ge (q/c)^{i-1}$ for $i>0$.
   This is trivial for $i=1$ and $u_{i+1}=\lfloor \frac qc u_i +1\rfloor \ge \lfloor (q/c)^i + 1 \rfloor \ge (q/c)^i$.

   Now define $E_i\in\mathbb Z$ such that $v_i = m - E_i$.
   We will show by induction that $0\le E_{i} \le u_{n-i}$ for $0\le i\le n$ such that $E_n = 0$, which gives the theorem.
   For a start $E_0=m\ge 0$ and $E_0 = \lfloor x/(q-c)\rfloor \le (q/c)^{n-1} \le u_n$.

   For $c\setminus q$ we can be slightly more specific, and support $x \le c (q/c)^n-c$.
   This follows by noting that $u_i = \frac{(q/c)^i-1}{q/c-1}$ for $i>0$, since all the $q/c$ terms are integral.
   Thus for $E_0=\floor{x/(q-c)}\le u_n$ it suffices to require $x\le c q^n-c$.

   For the induction step we plug in our expressions for $x$ and $v_i$:
   \begin{align*}
      v_{i+1}
      &= \left\lfloor \frac{(m-E_i+1)c+m(q-c)+h}{q}\right\rfloor
    \\&=
    m
    +
    \left\lfloor \frac{(- E_i+1)c +h}{q}\right\rfloor
    \\&=
    m
    - \left\lceil \frac{(E_i-1)c - h}{q}\right\rceil.
   \end{align*}
   The lower bound follows easily from $E_i \ge 0$ and $h\le q-c-1$:
   $$E_{i+1} = \left\lceil \frac{E_ic - h - c}{q}\right\rceil \ge
   \left\lceil \frac{- q + 1}{q}\right\rceil = 0.$$
   For the upper bound we use the inductive hypothesis as well as the bound $h\ge 0$:
   \begin{align*}
      E_{i+1}
      &= \left\lceil \frac{(E_i-1)c - h}{q}\right\rceil
    \\&\le\left\lceil (u_{n-i}-1)\frac{c}{q}\right\rceil
    \\&= \left\lceil \left\lfloor \frac{q}{c}u_{n-i-1} \right\rfloor \frac{c}{q}\right\rceil
    \\&\le \left\lceil u_{n-i-1}\right\rceil
    \\&= u_{n-i-1}.
   \end{align*}
   The last equality comes from $u_{n-i-1}$ being an integer.
   Having thus bounded the errors, the proof is complete.
\end{proof}
We can also note that if the algorithm is repeated more than $n$ times, the error stays at 0, since 
$\lceil (u_{n-i}-1)\frac{c}{q}\rceil = \lceil -\frac{c}{q}\rceil = 0$.

\subsection{Related Algorithms}

Modulus computation by Generalized Mersenne primes is widely used in the Cryptography community.
For example, four of the recommended primes in NIST's document ``Recommended Elliptic Curves for Federal Government Use''~\cite{nist} are Generalized Mersenne,
as well as primes used in the Internet Research Task Force's document ``Elliptic Curves for Security''~\cite{rfc7748}.
Naturally, much work has been done on making computations with those primes fast.
Articles like ``Simple Power Analysis on Fast Modular Reduction with Generalized Mersenne Prime for Elliptic Curve Cryptosystems''~\cite{sakai2006simple}
give very specific algorithms \emph{for each} of many well known such primes.
An example is shown in \Cref{alg:solina}.

\begin{algorithm}[H]
	\caption{Fast reduction modulo $p_{192} = 2^{192} - 2^{64} - 1$}
	\label{alg:solina}
	\begin{algorithmic}
		\State \textbf{input} $c \gets (c_5, c_4, c_3, c_2, c_1, c_0)$, where each $c_i$ is a 64-bit word, and $0 \le c < p^2_{192}$.
		\State $s_0 \gets (c_2, c_1, c_0)$
		\State $s_1 \gets (0, c_3, c_3)$
		\State $s_2 \gets (c_4, c_4, 0)$
		\State $s_3 \gets (c_5, c_5, c_5)$
		\State \textbf{return} $s_0 + s_1 + s_2 + s_3 \mod p_{192}$.
	\end{algorithmic}
\end{algorithm}

Division by Mersenne primes is a less common task, but a number of well known division algorithms can be specialized, such as
classical trial division, Montgomery's method and Barrett reduction.


The state of the art appears to be the modified Crandall Algorithm by Chung and Hasan~\cite{chung2006low}.
This algorithm, given in Algorithm \ref{alg:cch} modifies Crandall's algorithm~\cite{crandall1992method} from 1992 to compute division as well as modulo for generalized $2^b-c$ Mersenne primes.\footnote{
	Chung and Hasan also have an earlier, simpler algorithm from 2003~\cite{chung2003more},
	but it appears to give the wrong result for many simple cases.
	This appears in part to be because of a lack of the ``clean up'' while-loop at the end of Algorithm \ref{alg:cch}.
}
\begin{algorithm}[H]
	\caption{Crandall, Chung, Hassan algorithm. For $p=2^b-c$, computes $q, r$ such that $x = qp+r$ and $r<p$.}
	\label{alg:cch}
	\begin{algorithmic}
		\State $q_0 \gets x \rs n$
		\State $r_0 \gets x \andtt (2^b-1)$
		\State $q \gets q_0, r\gets r_0$
		\State $i \gets 0$
		\While{$q_i>0$}
		\State $t \gets q_i*c$
		\State $q_{i+1} \gets t \rs n$
		\State $r_{i+1} \gets t \andtt (2^b - 1)$
		\State $q\gets q+q_{i+1}$
		\State $r\gets r+r_{i+1}$
		\State $i\gets i+1$
		\EndWhile
		\State $t \gets 2^b-c$
		\While{$r\ge t$}
		\State $r\gets r-t$
		\State $q\gets q+1$
		\EndWhile
	\end{algorithmic}
\end{algorithm}
The authors state that for $2n+\ell$ bit input, Algorithm \ref{alg:cch}
requires at most $s$ iterations of the first loop, if $c < 2^{((s-1)n-\ell)/s}$.
This corresponds roughly to the requirement $x < 2^b (2^b/c)^s$, similar to ours.
Unfortunately, the algorithm ends up doing double work, by computing the quotient and remainder concurrently.
The algorithm also suffers from the extra while loop for ``cleaning up'' the computations after the main loop.
In practice, our method is 2-3 times faster. See \Cref{sec:experiments} for an empirical comparison.
%

\section{Experiments}\label{sec:experiments}
We perform experiments on fast implementations of Mersenne hashing (\Cref{alg:Mersenne}) and our Mersenne division algorithm (\Cref{alg:division-generalized}).
All code is available in our repository\\\href{https://github.com/thomasahle/mersenne/}{github.com/thomasahle/mersenne} and compiled with \texttt{gcc -O3}.

We tested \Cref{alg:Mersenne} against hashing over the finite fields $GF(2^{64})$ and $GF(2^{32})$.
The later is implemented, following Lemire~\cite{lemire2014strongly}, using the ``Carry-less multiplication' instruction, CLMUL, supported by AMD and Intel processors~\cite{GUERON2010549}.\footnote{
More precisely, given two $b$-bit numbers $\alpha = \sum_{i = 0}^{b - 1} \alpha_i 2^i$ and $\beta = \sum_{i = 0}^{b - 1} \beta_i 2^i$
the CLMUL instructions calculates $\gamma = \sum_{i = 0}^{2b - 2} \gamma_i 2^i$, where $\gamma_i = \bigoplus_{j = 0}^{j} \alpha_i \beta_{j - i}$.
If we view $\alpha$ and $\beta$ as elements in $GF(2)[x]$ then the CLMUL instruction corresponds to polynomial multiplication.
We can then calculate multiplication in a finite field, $GF(2^b)$, efficiently by noting that for any irreducible polynomial $p(x) \in GF(2)[x]$
of degree $b$ then $GF(2^b)$ is isomorphic to $GF(2)[x] / p(x)$. If we choose $p(x)$ such that the degree of $p(x) - 2^{b}$ is at
most $b/2$ then modulus $p(x)$ can be calculated using two CLMUL instructions.
For $GF(2^{64})$ we use the polynomial $p(x) = x^{64} + x^4 + x^3 + x + 1$ and for $GF(2^{32})$ we use the polynomial $p(x) = x^{32} + x^7 + x^6 + x^2 + 1$.
}
We hash a large number of $64$-bit keys into $[p]$ for $p=2^{89}-1$ using $k$-universal hashing for $k \in \{2, 4, 8\}$.
Since the intermediate values of our calculations take up to $64 + 89$ bits, all computations of \Cref{alg:Mersenne} are done with 128-bit output registers.
We perform the same experiment with $p=2^{61}-1$.
This allows us to do multiplications without splitting into multiple words, at the cost of a slightly shorter key range.

\begin{table}[H]
   \centering
   \begin{tabular}{c c}
      \begin{tabular}{r c | r r r}
          &  &  Mult- & Mersenne       & Carry-less \\
          $k$   & CPU & Shift & $p=2^{89}-1$ & GF($2^{64}$) \\
         \hline
           & a & \textbf{6.4} & 23.6 & 15.1 \\
         2 & b & \textbf{7.7} & 19.0 & 16.7 \\
           & c & \textbf{7.2} & 28.3 & 16.6 \\
         \hline
           & a & 99.3 & \textbf{65.7} & \textbf{65.7} \\
         4 & b & 157.7 &\textbf{68.7}  & \textbf{68.8} \\
           & c & 117.4 & 85.2 & \textbf{67.5}\\
         \hline
             & a & 1615.8 & \textbf{178.4} & 242.4 \\
         8 & b &  1642.3 & \textbf{187.4} & 246.8 \\
           & c & 1949.9 & 228.1 & \textbf{224.1} \\
      \end{tabular}
      &
      \begin{tabular}{r c | r r r}
          &  &  Mult- & Mersenne       & Carry-less \\
          $k$   & CPU & Shift & $p=2^{61}-1$ & GF($2^{32}$) \\
         \hline
             & a & \textbf{4.4} & 13.6 & 13.0 \\
         2 & b & \textbf{3.3} & 14.2 & 13.0 \\
           & c & \textbf{3.2} & 16.5 & 18.4 \\
         \hline
           & a & 57.6 & \textbf{31.6} & 60.3 \\
         4 & a & 54.6 & \textbf{34.3}  & 58.8 \\
           & c & 61.2 & \textbf{41.9} & 74.5 \\
         \hline
           & a & 650.7 & \textbf{88.0} & 218.7 \\
         8 & b & 635.6   & \textbf{88.0} & 212.0 \\
           & c & 750.8 & \textbf{127.8} & 253.2 \\
      \end{tabular}
   \end{tabular}
   \caption{Milliseconds for $10^7$ $k$-universal hashing operations.
      The standard deviation is less than $\pm1$ms.
      The three CPUs tested are
         a) Intel Core i7-8850H; 
         b) Intel Core i7-86650U; 
         c) Intel Xeon E5-2690 
         .
   }
   \label{tab:hashing-experiments}
\end{table}

The results in \Cref{tab:hashing-experiments} show that our methods outperform carry-less Multiplication for larger $k$, while being slower for $k=2$.
For $k=2$ the multiply-shift scheme~\cite{dietzfel96universal} is better yet, so in carry-less multiplication is nearly completely dominated.
For $k=4$, which we use for Count Sketch, the results are a toss-up for $p=2^{89}-1$, but the Mersenne primes are much faster for $p=2^{61}-1$.
We also note that our methods are more portable than carry-less, and we keep the two-for-one advantages described in the article.
\vspace{.5em} 

We tested \Cref{alg:division-generalized} against the state of the art modified Crandall's algorithm by Chung and Hasan (\Cref{alg:cch}), as well as the built-in truncated division algorithm in the GNU MultiPrecision Library, GMP~\cite{granlund2010gnu}.

\begin{table}[H]
   \centering
   \begin{tabular}{ c c }
      \begin{tabular}{ r | r r r }
         $b$ & Crandall & \Cref{alg:division-generalized} & GMP \\
         \hline
         32 & 396 & \textbf   {138}  & 149\\
         64 & 381 &   \textbf {142}  & 161\\
         128 & 564 &  \textbf {157}  & 239\\
         256 & 433 &  \textbf {187}  & 632\\
         512 & 687 &  \textbf {291}  & 1215\\
         1024 & 885 & \textbf {358}  & 2802
      \end{tabular}
      \hspace{.5em}
      &
      \hspace{.5em}
      \begin{tabular}{ r | r r r }
         $b$ & Crandall & \Cref{alg:division-generalized} & GMP\\
         \hline
         32 & 497 & 149 & \textbf{120}\\
         64 & 513 & 198 & \textbf{191}\\
         128 & 538 & \textbf{207} & 293\\
         256 & 571 & \textbf{222} & 603\\
         512 & 656 & \textbf{294} & 1167\\
         1024 & 786 & \textbf{372} & 2633\\
      \end{tabular}
   \end{tabular}
   \caption{Milliseconds for $10^7$ divisions of $2b$-bit numbers with $p=2^b-1$.
      The standard deviation is less than $\pm10$ms.
         On the left, Intel Core i7-8850H.
         On the right, Intel Xeon E5-2690 v4.
   }
   \label{tab:division-experiments}
\end{table}

The results in \Cref{tab:division-experiments} show that our method always outperforms the modified Crandall's algorithm, which itself outperforms GMP's division at larger bit-lengths.
At shorter bit-lengths it is mostly a toss-up between our method and GMP's.

We note that our code for this experiment is implemented entirely in GMP, which includes some overhead that could perhaps be avoided in an implementation closer to the metal.
This overhead is very visible when comparing \Cref{tab:hashing-experiments} and \Cref{tab:division-experiments}, suggesting that an optimized \Cref{alg:division-generalized} would beat GMP even at short bit-lengths.

\bibliographystyle{alpha}
\bibliography{general}

\end{document}